\documentclass[12pt,fleqn]{elsarticle}

\usepackage{amssymb}

\usepackage{amsfonts}
\usepackage{amssymb}
\usepackage{amsmath}
\usepackage{amsthm}

\usepackage[usenames]{color}
\usepackage{colortbl}

\usepackage{mathtext}
\usepackage[cp1251]{inputenc}
\usepackage[T2A]{fontenc}
\usepackage[czech,russian,ukrainian,english]{babel}

\newtheorem{lem}{Lemma}[section]

\newcommand{\e}{\mathrm{e}}
\renewcommand{\i}{\mathrm{i}}

\renewcommand{\d}{\mathrm{d}}
\newcommand{\ds}{\displaystyle}

\newcommand{\dsfrac}{\ds\frac}

\renewcommand{\(}{\left(}
\renewcommand{\)}{\right)}
\newcommand{\ol}{\overline}

\newenvironment{Proof}
    {\par\noindent{\bf Proof }}
    {\hfill$\scriptstyle\blacksquare$\vskip1cm}

\begin{document}

\begin{frontmatter}



\title{Riemann-Hilbert problem for Camassa-Holm equation with step-like initial data}


\author{Alexander Minakov}

\address{Mathematical Division, B. Verkin Institute for
Low Temperature Physics\\ and Engineering
of the National Academy of Science of Ukraine\\
47 Lenin Avenue, Kharkiv, 61103, Ukraine
\\
\vskip3mm
Doppler Institute for Mathematical Physics and Applied Mathematics,\\ Czech Technical University in Prague\\
Brehova 7, 11519 Prague, Czech Republic
\\
\vskip3mm Department of
Physics, Faculty of Nuclear Science and Physical
Engineering, \\Czech Technical University in Prague\\
Pohranicni 1288/1, D$\check{e}\check{c}\acute{i}$n, Czech
Republic\\\vskip3mm E-mail: minakov.ilt@gmail.com }

\begin{abstract}
The Cauchy problem for the Camassa -- Holm equation with step-like
initial conditions is reformulated as a Riemann -- Hilbert
problem. Then the initial value problem solution is obtained then
in a parametric form from the Riemann -- Hilbert problem solution.
\end{abstract}


\end{frontmatter}

\section{Introduction}
An inverse scattering approach, based on an appropriate Riemann --
Hilbert problem formulation, is developed for the step-like
initial value problem for the Camassa -- Holm (CH) equation on the
line, whose form is
\begin{equation}\label{CH}u_t-u_{txx}+2\omega
u_x+3uu_x=2u_xu_{xx}+uu_{xxx},\quad -\infty<x<\infty, \ t>0,
\end{equation}
\begin{equation}\label{init_cond}u(x,0)=u_0(x).
\end{equation}
Here $u_0(x)$ is a step-like function, that is $u_0(x)\to c_l$ as
$x\to-\infty$ and $u_0(x)\to c_r$ as $x\to+\infty,$ where $c_l$,
$c_r$ are some real constants. We consider real-valued classical
solutions $u(x,t)$ of the CH equation (\ref{CH}), which rapidly
tend to their limits as $x\rightarrow\pm\infty$, that is for any
$T\geq 0$
\[\max\limits_{0\leq t\leq
T}\int\limits_{-\infty}^{+\infty}\(1+|x|\)^{l+1}\times\]
\begin{equation}\label{conditions
u(x,t)}\times\(|m(x,t)-c_lH(-x)-c_rH(x)|+|m_x(x,t)|+|m_{xx}(x,t)|\)\d
x<\infty,\end{equation} where $l\geq0$ is some integer and
$H(x)=\left\{\begin{array}{l}1,\quad x\geq0,\\0,\quad
x<0\end{array}\right.$ is the Heaviside function.

 The Camassa-Holm equation describes the
unidirectional propagation of shallow water waves over a flat
bottom (R.~Camassa, D.~Holm, J.~Hyman, \cite {Camassa Holm 1993},
\cite{Camassa Holm Hyman 1994}) as well as axially symmetric waves
in a hyper-elastic rod (H.~Dai \cite{Dai}). Firstly it was found
using the method of recursion operators as a bi-Hamiltonian
equation with an infinite number of conserved functionals
(A.~Fokas, B.~Fuchssteiner \cite{Fokas Fuchssteiner 1981}).

For the case of vanishing initial data $c_l=c_r=0$ the
Riemann--Hilbert reformulation of CH equation and further
asymptotic analysis was done by D.Shepelsky with coauthors
\cite{Shep2007}, \cite{Shep_2008},\cite{Shep_2008_halh-line},
\cite{Shep_Its_2010} by transforming Lax pair into suitable for
asymptotical analysis form. Alternative approach based directly on
the scattering theory for the underlying Sturm -- Liouville
operator was developed by G.Teschl and A.Kostenko with coauthors
\cite{Kostenko_Shepelsky_Teschl_2009}.

Here we derive approach based on
\cite{Kostenko_Shepelsky_Teschl_2009}, but the approach based on
\cite{Shep2007} can be developed as well. We suppose that there
exists a classical solution of the Cauchy problem
(\ref{CH}),(\ref{init_cond}), and that this solution satisfies the
following condition for all values of time $t\geq0$:
\begin{equation}\label{m+omega} \dsfrac{m(x,t)+\omega}{c_r+\omega}>0,\end{equation} where
$$m(x,t):=u(x,t)-u_{xx}(x,t)$$ is the   so-called "momentum"
variable. Also we suppose that $\frac{c_l+\omega}{c_r+\omega}>0$.

The Camassa -- Holm equation (\ref{CH}), quantity
$\dsfrac{c_l+\omega}{c_r+\omega}$ and function
$\dsfrac{m(x,t)+\omega}{c_r+\omega}$ are invariant (see
\cite{Grunert_Holden_Raynaud_2011}, p.4) under the transformation
$$(\omega,\ u(x,t))\mapsto(\alpha\omega-\beta,\ v(x,t)=\alpha\,
u(x-\beta t,\alpha t)+\beta),$$ so there are only two
nonequivalent cases: $(c_r=0,\omega=1,c_l=:c>0)$ and
$(c_r=0,\omega=1,c_l\in(-1,0)).$ We restrict ourselves to the
first case, moreover, we will not assume that $\omega=1,$ so we
have $\(c_r=0, \dsfrac{c}{\omega}>0\).$

Our goal is to develop the inverse scattering approach to the CH
equation with step-like initial data, in view of its further
application for studying the long-time asymptotics. In section
\ref{sect Lax pair, Liouville transformation} we recall basic
facts about Camassa--Holm equation, introduce Jost solutions,
prove lemma \ref{lemma generic case} which guarantee boundedness
of the right transmission coefficient at the edge of spectrum at
the point $\dsfrac{\i}{2}\sqrt{\dsfrac{c}{c+\omega}}.$ In sections
\ref{Sect_RH_problem_1} and \ref{Sect_RH_problem_2} we state two
Riemann--Hilbert problems, for the right and left spectral
parameter, respectively. To deal with RH problem which is
continuous up to conjugation contour, in section
\ref{Sect_RH_problem_1} we assume $m(x,0)\leq x$ for all $x$. In
section \ref{Sect_RH_problem_2} to this end we suppose that $l=1$
in the formula (\ref{conditions u(x,t)}). In section \ref{sect
reconstruction of potential} Camassa-Holm solution is
reconstructed via RH problem solution. Two conservative laws for
step-like solutions of CH equation are obtained ((\ref{wp}),
(\ref{H_0[u]})).

\section{
Lax pair for CH equation, Liouville transformation, Jost solutions
and spectral funcions}\label{sect Lax pair, Liouville
transformation}
 In this section we derive a vector Riemann -- Hilbert problem
directly from the scattering theory for the Sturm -- Liouville
operator with a step-like potential \cite{Cohen_Kappeler},
\cite{Bazargan_2008}. We begin by recalling some required results
for the Camassa -- Holm equation
\cite{Constantin_2001},\cite{Constantin_Gerdikov_Ivanov},
\cite{Shep2007},\cite{Kostenko_Shepelsky_Teschl_2009}.

The starting point for our considerations is the Lax
representation: the CH equation is the compatibility condition of
two linear equations

\begin{subequations}\label{Lax_representation_CH}
\begin{eqnarray}\label{Shturm_x-eq}
\dsfrac{\omega}{m+\omega}\(-\,\varphi''_{xx}+\dsfrac{1}{4}\
\varphi\)=\lambda\,\varphi\ ,
\\\label{Shturm_t-eq}
\varphi_t=-\(\dsfrac{\omega}{2\lambda}+u\)\varphi'_x+\frac{u_x}{2}\varphi\
.
\end{eqnarray}
\end{subequations}
\noindent Assuming that $\frac{m+\omega}{\omega}\geq 0,$ equation
(\ref{CH}) can be equivalently  written as
\begin{equation}\label{CH alternative}
\(\sqrt{\frac{m+\omega}{\omega}}\)_t=-\(u\sqrt{\frac{m+\omega}{\omega}}\)_x\
.
\end{equation}

 \noindent Introduce Liouville transform

\begin{equation}\label{y_x}
y=y(x,t)\equiv x-\int\limits_x^{+\infty}\(\sqrt{\dsfrac{m(\tilde
x,t)+\omega}{\omega}\ }-1\)\d \tilde x,
\end{equation}
\begin{equation}\label{psi_varphi}
\psi(y)=\sqrt[4]{\dsfrac{m+\omega}{\omega}\ }\ \varphi(x).
\end{equation}

\noindent There is a conservative law
\begin{equation}\nonumber\wp=x\(\sqrt{\dsfrac{c+\omega}{\omega}}-1\)-c\,\sqrt{\dsfrac{c+\omega}{\omega}\ }\
t+\int\limits_{-\infty}^x\(\sqrt{\dsfrac{m(\xi,t)+\omega}{\omega}}-\sqrt{\dsfrac{c+\omega}{\omega}}\)\d\xi+
\end{equation}\begin{equation}\label{wp}\hfill+\int\limits_x^{+\infty}\(\sqrt{\dsfrac{m(\xi,t)+\omega}{\omega}}-1\)\d\xi
\end{equation}
which can be checked directly.

\noindent Due to the conservative law (\ref{wp}) $y$ can be also
expressed in the following way:
\begin{equation}\label{y_x_-}
y=x-\int\limits_x^{+\infty}\(\sqrt{\dsfrac{m+\omega}{\omega}\
}-1\)\d
r=\end{equation}\[=-\wp+\sqrt{\dsfrac{c+\omega}{\omega}}\left[x-ct+\int\limits_{-\infty}^x\(\sqrt{\dsfrac{m+\omega}{c+\omega}}-1\)\d
r\right],
\]
so $y\to\pm\infty$ as $x\to\pm\infty$ and vise versa.

Spectral problem (\ref{Shturm_x-eq}) now reads as follows:

\begin{equation}\label{Shturm_psi_y-eq}
-\psi''_{yy}(y,t;k)+v(y,t)\psi(y,t;k)=k^2\psi(y,t;k),
\end{equation}
or, equivalently,
\begin{equation}\label{Shturm_psi_y-eq_for z}
-\psi''_{yy}(y,t;k)+\(v(y,t)+\frac{c}{4(c+\omega)}\)\psi(y,t;k)=z^2\psi(y,t;k).
\end{equation}

Here
\begin{equation}\label{lambda_k}
\lambda=:k^2+\frac{1}{4}=:z^2+\frac{\omega}{4(c+\omega)},
\end{equation}
and
\begin{equation}\label{vqT} v(y,t)=-
\dsfrac{m}{4(m+\omega)}+\dsfrac{\omega}{4}\dsfrac{m_{xx}}{\(m+\omega\)^2}
-\dsfrac{5\omega}{16}\dsfrac{m_x^2}{(m+\omega)^3},
\end{equation}

\noindent so $v(y,t)\to0$ as $y\to+\infty,$ and
$v(y,t)\to-\dsfrac{c}{4(c+\omega)}$ as $y\to-\infty.$ From our
assumption (\ref{conditions u(x,t)}) we get that
$v(y,t)+\frac{c}{4(c+\omega)}\in L_1(\mathbb{R}_-, (1+|y|)\d y)$
and $v(y,t)\in L_1(\mathbb{R}_+, (1+|y|)\d y).$ We will consider
$z$ as a function of $k,$
$$z=z(k)=\sqrt{k^2+\frac{c}{4(c+\omega)}},$$ where the branch of the cut
across the segment $\(\i\ \sqrt{\frac{c}{4(c+\omega)}},-\i\
\sqrt{\frac{c}{4(c+\omega)}}\)$ is fixed by the condition $z\sim
k\ $ as $\ k\rightarrow\infty.$

Spectral problems (\ref{Shturm_psi_y-eq}),
(\ref{Shturm_psi_y-eq_for z}) are well studied (see
{\cite{Kappeler_1985}}, \cite{Bazargan_2008},
\cite{Egorova_2008}), and all results known for them can be
readily applied.

\begin{lem}\label{lemma Jost solutions_Sturm-Liouville}
There exist two Jost solutions $\varphi_{\pm}(x,t;k)$ which solve
the system of differential equations (\ref{Lax_representation_CH})
and satisfy
\begin{equation}\label{asymptotics varphi - x}
\lim\limits_{x\rightarrow-\infty}\sqrt[4]{\frac{c+\omega}{\omega}}\exp\left\{\i
z\sqrt{\frac{c+\omega}{\omega}}\(x-\(c+\frac{\omega}{2\lambda}\)t\)\right\}\varphi_-(x,t;k)=1,
\end{equation}
\begin{equation}\label{asymptotics varphi + x}
\lim\limits_{x\rightarrow+\infty}\e^{-\i kx+\frac{\i\omega
kt}{2\lambda}}\varphi_+(x,t;k)=1.
\end{equation}
Function $\ \exp\left\{\frac{-\i zt\sqrt{(c+\omega)\omega\
}}{2\lambda}\right\}\varphi_-(x,t;k)\ $ is analytic for $\Im
z(k)>0$ and continuous for $\Im z(k)\geq0$,\\
function $\ \exp\left\{\frac{\i
kt\omega}{2\lambda}\right\}\varphi_+(x,t;k)\ $ is analytic for
$\Im k>0$ and continuous for $\Im k\geq0.$ As
$k\rightarrow\infty$, $\Im k\geq0,$
we have\\\\
$\varphi_-(x,t;k)=\sqrt[4]{\dsfrac{\omega}{m+\omega}}\ \cdot
$\begin{equation}\nonumber\cdot\exp\left\{-\i z
\sqrt{\frac{c+\omega}{\omega}}\(x+\int\limits_{-\infty}^x\(\sqrt{\frac{m+\omega}{c+\omega}}-1\)\d\tilde
x-\(c+\frac{\omega}{2\lambda}\)t\)\right\}\cdot
\end{equation}
\begin{equation}\label{asymptotics varphi - k}\cdot\(1-\frac{1}{2\i k}\int\limits_{-\infty}^y\(v(\tilde y,t)+\frac{c}{4(c+\omega)}\)
\d\tilde y+\mathrm{O}\(k^{-2}\)\),
\end{equation}

\begin{equation}\label{asymptotics varphi + k}
\varphi_+(x,t;k)=\sqrt[4]{\frac{\omega}{m+\omega}} \exp\left\{\i
k\(x-\int\limits_{x}^{+\infty}\(\sqrt{\frac{m+\omega}{\omega}}-1\)\d\tilde
x\)-\frac{\i\omega kt}{2\lambda}\right\}\cdot \end{equation}
\begin{equation}\cdot
\(1-\frac{1}{2\i k}\int\limits^{+\infty}_y v(\tilde y,t) \d\tilde
y+\mathrm{O}\(k^{-2}\)\).
\end{equation}

\noindent Moreover, the following relations are satisfied:\\
$\varphi_{\pm}(x,t;k)=\overline{\varphi_{\pm}(x,t;-\overline{k})},$\\
$\varphi_-(x,t;k\pm0)=\overline{\varphi_-(x,t;\overline{k}\mp0)},\quad
k\in\(\i\ \sqrt{\dsfrac{c}{4(c+\omega)}}\ ,\ -\i\
\sqrt{\dsfrac{c}{4(c+\omega)}}\);$
\\\\\\
$\varphi_+(x,t;k)=\overline{\varphi_+(x,t;k)},\quad k\in\(\i\
\sqrt{\dsfrac{c}{4(c+\omega)}}\ ,0\);$
\\
$\varphi_-(x,t;k\pm0)=\varphi_-(x,t;\overline{k}\pm0),\quad
k\in\(\i\ \sqrt{\dsfrac{c}{4(c+\omega)}}\ ,\ -\i\
\sqrt{\dsfrac{c}{4(c+\omega)}}\).$
\end{lem}
\begin{proof}
This is straightforward from the corresponding results for the
spectral problem (\ref{Shturm_psi_y-eq}) (see
\cite{Cohen_Kappeler}) by virtue of the Liouville transform
(\ref{y_x}), (\ref{psi_varphi}). Just notice
\begin{equation}\label{Jost solutions and Kappeler Jost solutions}
\begin{array}{ccc}\varphi_-(x,t;k)=\sqrt[4]{\frac{\omega}{m+\omega}}\exp\left\{-\i\wp
z+\frac{\i\sqrt{(c+\omega)\omega\ }\ zt}{2\lambda}\right\}
\psi^K_-(y,t;k),\\
\varphi_+(x,t;k)=\sqrt[4]{\frac{\omega}{m+\omega}}\ \
\e^{\frac{-\i\omega kt}{2\lambda}} \
\psi^K_+(y,t;k),\end{array}\end{equation}
 where $\wp$ is a conserved quantity
of CH equation (\ref{wp}) and $\psi^K_{\pm}(y,t;k)$ are the Jost
solutions of (\ref{Shturm_psi_y-eq}) determined by their
aymptotics
\begin{equation}\label{Kappeler Jost solutions}\begin{array}{ccc}
\e^{+\i z y}\ \psi^K_-(y,t;k)\rightarrow1\ \textrm { as }\
y\rightarrow-\infty, \quad \Im z\geq0\\
\e^{-\i k y}\ \psi^K_+(y,t;k)\rightarrow1\ \textrm { as }\
y\rightarrow+\infty,\quad \Im k\geq0.
\end{array}
\end{equation}

\end{proof}

%
%
%



\begin{lem} Jost solutions $\varphi_{\pm}(x,t;k)$
determined by (\ref{asymptotics varphi - x}), (\ref{asymptotics
varphi + x})
 satisfy $t$-equation
(\ref{Shturm_t-eq}).
\end{lem}

\begin{proof}

%

We follow the well-known idea from \cite{Marchenko_1972_1986},
chapter 4, paragraph 2. First the statement is to be proved for
real $k$ and then can be extended by analyticity to $\Im k\geq0$,
$\Im z\geq0$, respectively. Let $\varphi(x,t;k)$ be one of the
functions $\varphi_{\pm}(x,t;k).$
It is straightforward that if $\varphi(x,t;k)$ satisfy (\ref{CH})
and (\ref{Shturm_x-eq}), then
\begin{equation}\nonumber
\(\varphi_t+\(\frac{\omega}{2\lambda}+u\)\varphi_x-\frac{u_x}{2}\,\varphi\)_{xx}=
\end{equation}
\begin{equation}\nonumber\hskip4cm=\(\dsfrac{1}{4}-\dsfrac{\lambda(m+\omega)}{\omega}\)
\(\varphi_t+\(\frac{\omega}{2\lambda}+u\)\varphi_x-\frac{u_x}{2}\,\varphi\),
\end{equation}
i.e.
$\varphi_t+\(\frac{\omega}{2\lambda}+u\)\varphi_x-\frac{u_x}{2}\,\varphi$
satisfies the same differential equation (\ref{Shturm_x-eq}) as
$\varphi$. If we define by $\hat\varphi$ another solution of
(\ref{Shturm_x-eq}) such that $\varphi$ and $\hat\varphi$ are
linear independent, then there exist independent on $x$ functions
$A(t;k)$, $B(t;k)$ such that
\begin{equation}\nonumber
\varphi_t(x,t;k)+\(\frac{\omega}{2\lambda}+u(x,t)\)\varphi_x(x,t;k)-\frac{u_x(x,t)}{2}\,\varphi(x,t;k)=
\end{equation}
\begin{equation}\nonumber \qquad\qquad\qquad=A(t;k)\varphi(x,t;k)+ B(t;k)\hat\varphi(x,t;k).
\end{equation}

\noindent For $\psi(y,t;k)=\sqrt[4]{\dsfrac{m+\omega}{\omega}\ }\
\varphi(x,t;k), $ $\quad
\hat\psi(y,t;k)=\sqrt[4]{\dsfrac{m+\omega}{\omega}\ }\
\hat\varphi(x,t;k)$ \\the last equation reads as
\begin{equation}\nonumber
\psi_t+\dsfrac{\omega}{2\lambda}\sqrt{\dsfrac{m+\omega}{\omega}}\psi_y-
\(\dsfrac{um_x}{4(m+\omega)}+\dsfrac{m_t}{4(m+\omega)}+\dsfrac{u_x}{2}+\dsfrac{\omega
m_x}{8\lambda(m+\omega)}\)\psi=\end{equation}
\begin{equation}\label{Vypolnenie_t_eq}\hskip10cm =A\psi+B\hat\psi.
\end{equation}

a) Now let us take $\varphi(x,t;k)=\varphi_+(x,t;k)$ and
$\hat\varphi(x,t;k)=\overline{\varphi_+(x,t;k)}.$ Function
$\psi(x,t;k)$ defined by (\ref{psi_varphi}) has the following
asymptotic behavior as $x\to+\infty$ for $\Im k=0:$
\begin{equation}\nonumber
\e^{-\i ky+\frac{2\,\i\, \omega kt}{4k^2+1}}\
\psi_+(y,t;k)\rightarrow1,\quad y\to+\infty,\quad \mathrm{Im\ }
k=0,
\end{equation}
\begin{equation}\nonumber
\e^{-\i ky+\frac{2\,\i\, \omega kt}{4k^2+1}}\
\psi'_{+,y}(y,t;k)\rightarrow\i k,\quad y\to+\infty,\quad
\mathrm{Im\ } k=0,
\end{equation}
\begin{equation}\nonumber
\e^{-\i ky+\frac{2\,\i\, \omega kt}{4k^2+1}}\
\psi'_{+,t}(y,t;k)\rightarrow\dsfrac{-2\i k\omega}{4k^2+1},\quad
y\to+\infty,\quad \mathrm{Im\ } k=0.
\end{equation}

Substituting these asymptotics into (\ref{Vypolnenie_t_eq}) and
taking the limit as $y\to+\infty$ we get 0 in the left-hand-side,
which is possible only if $A(t;k)\equiv0$, $B(t;k)\equiv0.$

b) Now let us take $\varphi(x,t;k)=\varphi_-(x,t;k)$ and
$\hat\varphi(x,t;k)=\overline{\varphi_-(x,t;k)}.$ Function
$\psi(y,t;k)$ defined (\ref{psi_varphi}) has the following
asymptotic behavior as $x\to+\infty$ for $\Im k=0:$
\begin{equation}\nonumber
\e^{\i \(y+\wp\)z(k)-\frac{2\, \i\, \,\sqrt{(c+\omega)\omega}\
t\,z(k)}{4k^2+1}}\ \psi_-(y,t;k)\rightarrow1,\quad
y\to-\infty,\quad \mathrm{Im\ }z(k)=0.
\end{equation}
\begin{equation}\nonumber
\e^{\i \(y+\wp\)z(k)-\frac{2\, \i\, \,\sqrt{(c+\omega)\omega}\
t\,z(k)}{4k^2+1}}\ \psi_{-,y}(y,t;k)\rightarrow-\i z(k),\
y\to-\infty,\ \mathrm{Im\ }z(k)=0.
\end{equation}
\begin{equation}\nonumber
\e^{\i \(y+\wp\)z(k)-\frac{2\, \i\, \,\sqrt{(c+\omega)\omega}\
t\,z(k)}{4k^2+1}}\
\psi_{-,t}(y,t;k)\rightarrow\frac{2\i\sqrt{(c+\omega)\omega\,}\,z(k)}{4k^2+1},\quad
y\to-\infty,\end{equation}\begin{equation}\nonumber\quad\qquad\qquad\qquad\qquad\qquad\qquad\qquad\qquad\qquad\qquad\qquad\qquad\qquad
\mathrm{Im\ }z(k)=0.
\end{equation}

Substituting these asymptotics into (\ref{Vypolnenie_t_eq}) and
taking the limit as $y\to-\infty$ we get 0 in the left-hand-side,
which is possible only if $A(t;k)\equiv0$, $B(t;k)\equiv0.$

\end{proof}

Next, we have the scattering relations
\begin{subequations}\label{scattering relations varphi}
\begin{eqnarray}
\varphi_-(x,t;k)=a_+(k)\
\overline{\varphi_+(x,t;\overline{k})}+b_+(k)\
\varphi_+(x,t;k),\quad k\in\mathbb{R}\setminus\left\{k=0\right\},
\nonumber\\\\
\varphi_+(x,t;k)=a_-(k)\
\overline{\varphi_-(x,t;\overline{k})}+b_-(k)\
\varphi_-(x,t;k),\quad z(k)\in\mathbb{R},\Im k\geq 0,\nonumber\\
\end{eqnarray}
\end{subequations}
where $a_{\pm}^{-1}(k),$ $\frac{b_{\pm}(k)}{a_{\pm}(k)}$ are the
transmission and reflection coefficients respectively. Introduce
also Wronskian
\begin{equation}\label{Wronskian definition}W(k):=W\left\{\varphi_-(x,t;k),\varphi_+(x,t;k)\right\}=\varphi_-\cdot\frac{\partial\varphi_+}{\partial
x}-\varphi_+\cdot\frac{\partial\varphi_-}{\partial x}.
\end{equation} We need the following preliminary lemmas:

\begin{lem}\label{symmetries a b}
Spectral functions $a_{\pm}(k),$ $\ b_{\pm}(k)$ have extended
domains of definition and are expressed in terms of
$\varphi_{\pm}(x,t;k)$ as follows:
\begin{enumerate}
\item $W\left\{\varphi_-(x,t;k),\ \varphi_+(x,t;k)\right\}=2\i k\,
a_+(k),\quad \Im z(k)\geq0;$ \item
$W\left\{\overline{\varphi_+(x,t;\overline{k})},\
\varphi_-(x,t;k)\right\}=2\i k\, b_+(k),\quad
z(k)\in\mathbb{R}\cap \Im k\leq0;$ \item
$W\left\{\varphi_-(x,t;k),\ \varphi_+(x,t;k)\right\}=2\i z(k)\,
a_-(k),\quad \Im z(k)\geq0;$ \item $W\left\{\varphi_+(x,t;k),\
\overline{\varphi_-(x,t;\overline{k})}\right\}=2\i z(k)\,
b_-(k),\quad z(k)\in\mathbb{R}\cap\Im k\geq0.$
\end{enumerate}
They possess the following properties:

\renewcommand{\theenumi}{\Roman{enumi}}
\begin{enumerate}
\item $\overline{a_+(-\overline{k})}=a_+(k),\quad
\Im z(k)\geq 0;$\\
$\qquad \overline{b_+(-\overline{k})}=b_+(k),\quad
z(k)\in\mathbb{R}\cap\Im k\leq0;$ \item
$\overline{a_-(-\overline{k})}=a_-(k),\quad
k\in\overline{\mathbb{C}^{\,c}_+};\qquad
\overline{b_-(-\overline{k})}=b_-(k),\quad
z(k)\in\mathbb{R}\cap\Im k\geq0 ;$ \item
$k\,a_+(k)=z(k)a_-(k),\quad
\Im z(k)\geq0;$\\
$k\,\overline{b_+(\overline{k})}=-z(k)b_-(k),\quad
z(k)\in\mathbb{R}\cap\Im k\geq0;$ \item
$\overline{b_+(\overline{k}\pm0)}=a_{+}(k\mp0),\quad
b_-(k\pm0)=a_-(k\mp0),\\
b_-(k\pm0)=\overline{a_-(k\pm0)},\quad
b_+(\overline{k}\pm0)=a_+(k\pm0),\quad
k\in\left[0,\i\sqrt{\dsfrac{c}{4(c+\omega)}}\right];$
\item $a_+(k)b_-(k)+a_-(k)\overline{b_+(\overline{k})}=0,\quad z(k)\in\mathbb{R}\cap\Im k\geq0,$\\\\
$a_+(k)\overline{b_-(\overline{k})}+a_-(k)b_+(k)=0,\quad k\in\mathbb{R}\setminus\left\{0\right\},$\\\\
$\dsfrac{k}{z(k)}\(|a_+(k)|^2-|b_+(k)|^2\)=1,\quad
k\in\mathbb{R}\setminus\left\{0\right\};$\\\\
$\dsfrac{z(k)}{k}\(|a_-(k)|^2-|b_-(k)|^2\)=1,\quad
k\in\mathbb{R}\setminus\left\{0\right\}.$
\end{enumerate}
\renewcommand{\theenumi}{\arabic{enumi}}
\end{lem}

\begin{lem}\label{lemma Wronskian}
\noindent Function $W(k)$ (\ref{Wronskian definition}) possesses
the following properties:
\begin{itemize}
\item $W(k)$ is analytic in $\Im z(k)>0$ and continuous in $\Im
z(k)\geq0;$ \item $W(k)\neq0$ for any $k\in\left\{k:
z(k)\in\mathbb{R}\wedge\Im k\geq0\wedge
k\neq\i\sqrt{\dsfrac{c}{4(c+\omega)}}\right\}.$ \item zeros of
$W(k)$ in the domain $\Im z(k)>0$ are simple and lie in the
interval $\(\i\sqrt{\dsfrac{c}{4(c+\omega)}},\ \dsfrac{\i}{2}\).$
\end{itemize}
\end{lem}
\begin{proof}
It is sufficient to prove only the third property, as first two
are immediate from corresponding results from
\cite{Kappeler_1985}. The fact that the discrete spectrum may lie
only in the interval $\(\i\sqrt{\frac{c}{4(c+\omega)}},+\i\
\infty\)$, also follows from \cite{Kappeler_1985}. The fact that
it indeed can lie only in the interval
$\(\i\sqrt{\frac{c}{4(c+\omega)}},\ \frac{\i}{2}\),$ can be proved
as in \cite[Claim 2, p. 962]{Constantin_2001}.
\end{proof}

The following lemma establishes a condition on an initial data
which provides $W\(\i\sqrt{\frac{c}{4(c+\omega)}}\)\neq0.$

\begin{lem}\label{lemma generic case}
Suppose that for all $x\in\mathbb{R}$
$$\dsfrac{m(x,0)-c}{c+\omega}\leq0.$$
Then $\qquad W\(\i\ \sqrt{\frac{c}{4(c+\omega)}}\)\neq0.$
\end{lem}
\begin{proof}
Define $k_0=\i\ \sqrt{\frac{c}{4(c+\omega)}},$
$\lambda_0=\lambda(k_0),$ and suppose that $W\(k_0\)=0.$ In that
case the Jost solutions $\varphi_-(x,t;k_0)$ and
$\varphi_+(x,t;k_0)$ are linearly dependent, that is there exists
a constant $C\in\mathbb{C}\setminus\left\{0\right\}$ such that
$\varphi_-(x,t;k_0)=C\varphi_+(x,t;k_0).$ From (\ref{asymptotics
varphi - x}), (\ref{asymptotics varphi + x}) we conclude that
function \\$\varphi(x,t):=\exp\left\{-\i
z\sqrt{\frac{c+\omega}{\omega}}\(c+\frac{\omega}{2\lambda}\)t\right\}\varphi_-(x,t;k_0)$
possesses the following properties (notice that $z(k_0)=0$):
\[\sqrt[4]{\dsfrac{m+\omega}{\omega}}\varphi(x,t)\rightarrow1,\quad \varphi_x(x,t)\rightarrow0, \quad\textrm{ as } x\rightarrow-\infty,\]
\[\varphi(x,t)\rightarrow0,
\quad\varphi_x(x,t)\rightarrow0,\quad \textrm{ as }
x\rightarrow+\infty.\]

\noindent Equation (\ref{Lax_representation_CH}a) can be rewritten
as follows:
\[-\varphi_{xx}(x,t)+\frac{1}{4}\varphi(x,t)=\lambda_0\ \frac{m+\omega}{\omega}\varphi(x,t),\]
where $\lambda_0=k_0^2+\frac{1}{4}=\frac{\omega}{4(c+\omega)}.$
Multiply the last equation by $\overline{\varphi}$ and integrate
from $R$ to $+\infty:$
\[\int\limits_{R}^{+\infty}\hskip-2mm-\varphi_{xx}(x,t)\overline{\varphi(x,t)}\d x+
\frac{1}{4}\int\limits_{R}^{+\infty}\hskip-2mm|\varphi(x,t)|^2\d
x\hskip-1mm=\hskip-1mm\lambda_0\int\limits_{R}^{+\infty}\hskip-2mm\frac{m+\omega}{\omega}|\varphi(x,t)|^2\d
x.\] By integrating by parts and transferring of a summand to the
right side of the equation, we get
\[-\varphi_x\overline{\varphi}\left|_{R}^{+\infty}\right.+\int\limits_{R}^{+\infty}|\varphi_{x}(x,t)|^2\d x
=\int\limits_{R}^{+\infty}\frac{m-c}{4(c+\omega)}|\varphi(x,t)|^2\d
x.\] Here we can take the limit as $R\rightarrow-\infty,$ and so
come to the equation
\[\int\limits_{-\infty}^{+\infty}|\varphi_{x}(x,t)|^2\d x
=\int\limits_{-\infty}^{+\infty}\frac{m(x,t)-c}{4(c+\omega)}|\varphi(x,t)|^2\d
x.\] As the left-hand side is strictly positive, and the
right-hand side is non-positive for $t=0$, we come to
contradiction with the statement of the lemma. So, $W(k_0)\neq0.$
\end{proof}

\begin{lem}\label{lemma properties a(k) b(k)}
The transmission coefficients $a^{-1}_{\pm}$ are meromorphic for
$\Im z(k)>0$ with simple poles at $\i\kappa_1,...,\i\kappa_N,$
$\sqrt{\frac{c}{4(c+\omega)}}<\kappa_N<...<\kappa_1<\dsfrac{1}{2},$
and continuous for $k\in\left\{k:\Im z(k)\geq0\wedge \Im k\geq0
\wedge z(k)\neq0\right\}.$ Moreover, if
\begin{enumerate}
\item we take $l=0$ in (\ref{conditions u(x,t)}) and require that
$\forall x\in\mathbb{R}: \quad c>m(x,0),$ \\then $a_+^{-1}$ is
continuous up to the point $z(k)=0;$ \item we take $l=1$ in
(\ref{conditions u(x,t)}), then $a_-^{-1}$ is continuous up to the
point $z(k)=0.$
\end{enumerate}
Asymptotically as $k\rightarrow\infty$ we have
\begin{equation}\label{transmission coeff+- asymptotics}
a^{-1}_{\pm}(k)=\e^{\i\wp k}\(1+\mathrm{O}(k^{-1})\).
\end{equation}
The residues of $a_+^{-1}(k)$ and $a_-^{-1}(k(z))$ are given by
\begin{equation}\label{transmission coeff+-residues}
\mathrm{Res}_{\, \i\kappa_j}\
\dsfrac{1}{a_+(k)}=\i\mu_j\gamma_{+,j}^2\  =\ \mathrm{Res}_{\,
z(\i\kappa_j)}\ \dsfrac{1}{a_-(k(z))}=\i\mu_j^{-1}\gamma_{-,j}^2
\end{equation}
where
\begin{equation}\label{gamma_+-j}
\gamma_{\pm,j}^{-2}:=\int\limits_{-\infty}^{+\infty}
\(\varphi_{\pm}(x,t;\i\kappa_j)\)^2\dsfrac{m+\omega}{\omega}\d x>0
\end{equation}
and $\varphi_+(x,t;\i\kappa_j)=\mu_j\varphi_-(x,t;\i\kappa_j)$
with quantities $\gamma_{\pm,j}$, $\mu_j$ independent on $t.$
\end{lem}

Note that if $a_{\pm}^K(k,t)$, $b_{\pm}^K(k,t)$,
$\gamma^K_{\pm,j}(t)$, $\mu_j^K(t)$, \[W^K(t;k)\equiv
W\left\{\psi^K_-(y,t;k),\psi^K_+(y,t;k)\right\}\equiv\psi^K_-\cdot\frac{\partial\psi^K_+}{\partial
y}-\psi^K_+\cdot\frac{\partial\psi^K_-}{\partial y}\]are the
corresponding quantities for (\ref{Shturm_psi_y-eq}), then
\[a_{\pm}(k)=a_{\pm}^K(k,t)\exp\left\{-\i\wp z+\frac{\i\sqrt{(c+\omega)\omega\ }\ zt}{2\lambda}-\frac{\i\omega kt}{2\lambda}\right\},\]
\[b_{\pm}(k)=b_{\pm}^K(k,t)\exp\left\{\pm\(-\i\wp z+\frac{\i\sqrt{(c+\omega)\omega\ }\ zt}{2\lambda}+\frac{\i\omega kt}{2\lambda}\)\right\},\]
\[\gamma_{+,j}=\gamma_{+,j}^K(t)\exp\left\{\frac{\i\omega k_j t}{2\lambda_j}\right\},\]
\[\gamma_{-,j}=\gamma_{-,j}^K(t)\exp\left\{\i\wp z_j-\frac{\i\sqrt{\(c+\omega\)\omega\ } z_j t}{2\lambda_j}\right\},\]
\[\mu_j=\mu_j^K(t)\exp\left\{\i\wp z_j-\frac{\i\sqrt{\(c+\omega\)\omega\ }z_j t}{2\lambda_j}-\frac{\i\omega k_j t}{2\lambda_j}\right\}\]
with $\lambda_j:=\lambda(k_j)$, $z_j:=z(k_j),$
\[
W(k)=\exp\left\{-\i\wp z+\frac{\i\sqrt{(c+\omega)\omega\
}zt}{2\lambda}-\frac{\i\omega t k}{2\lambda}\right\}W^K(t;k)
\]
 and hence all
results known for (\ref{Shturm_psi_y-eq}) are easily applied in
our situation. Particularly, as
$$\frac{\i\sqrt{(c+\omega)\omega\ }\ zt}{2\lambda}-\frac{\i\omega
kt}{2\lambda}=\frac{\i tc}{2\(z\sqrt{\frac{c}{\omega}+1\ }\
+k\)},$$ then $a_{\pm}(k)$ are regular at the point
$k=\dsfrac{\i}{2}.$

\section{Vector Riemann -- Hilbert problem
(1).}\label{Sect_RH_problem_1} Suppose that the condition of lemma
{\ref{lemma generic case}} is satisfied. We define a vector
Riemann -- Hilbert problem as follows: sectionally meromorphic
function
$V_{\mathfrak{r}}(x,t;k)=\(V_{\mathfrak{r},\,1}(x,t;k),V_{\mathfrak{r},\,2}(x,t;k)\)$
is defined by

\begin{equation}\label{M_psi}
\left\{\begin{array}{l}
\sqrt[4]{\frac{m+\omega}{\omega}}\(\dsfrac{1}{a_+(k)}\varphi_-(x,t;k)\,\e^{\i
g(y,t;k)}\ , \ \varphi_+(x,t;k)\,\e^{-\i g(y,t;k)}\), \quad \Im
z(k)>0,
\\
\sqrt[4]{\frac{m+\omega}{\omega}}\(\overline{\varphi_+(x,t;\overline{k})}\,\e^{\i
g(y,t;k)}\ , \
\dsfrac{1}{\overline{a_+(\overline{k})}}\overline{\varphi_-(x,t;\overline{k})}\,\e^{-\i
g(y,t;k)}\),\quad \Im z(k)<0,
\end{array}\right.
\end{equation}
\noindent where $g(y,t;k)=k y-\frac{2\omega k t}{4k^2+1}.$

We are interested in the jump relations of
$V_{\mathfrak{r}}(x,t;k)$ on the contour
$\Sigma_{\mathfrak{r}}=\mathbb{R}\cup
\left[\i\sqrt{\frac{c}{4(c+\omega)}}\
,-\i\sqrt{\frac{c}{4(c+\omega)}}\right]$. The orientation of the
contour is chosen as follows: from $-\infty$ to $+\infty$ and from
$+\i\sqrt{\frac{c}{4(c+\omega)}}$ to
$-\i\sqrt{\frac{c}{4(c+\omega)}}.$ Positive side of the contour is
on the left, negative is on the right. By $M_{\pm}(x,t;k)$ we
denote the limit from the positive/negative side of the contour.

The scattering relations (\ref{scattering relations varphi}) and
lemma \ref{lemma properties a(k) b(k)} provides the following
properties of function $V_{\mathfrak{r}}(x,t;k):$

\begin{enumerate}
\item The analyticity: \\$V_{\mathfrak{r}}(x,t;.)$ is meromorphic
in $\mathbb{C}\setminus\Sigma_{\mathfrak{r}}$ with simple poles at
$\pm\i\kappa_j$ and continuous up to the boundary;

\item The jump relations: for $k\in\Sigma_{\mathfrak{r}}$\\
$V_{\mathfrak{r}}^-(x,t;k)=V_{\mathfrak{r}}^+(x,t;k)J_{\mathfrak{r}}(x,t;k),$
where
\begin{equation}\label{J_R}
J_{\mathfrak{r}}(x,t;k)=\begin{pmatrix}1&\frac{\overline{b_+(k)}}{\overline{a_+(k)}}\,\,\e^{-2\i
g(y,t;k)}\\-\frac{b_+(k)}{a_+(k)}\,\e^{2\i
g(y,t;k)}&\frac{z(k)}{k\,|a_+(k)|^2}\end{pmatrix},\quad
k\in\mathbb{R}\setminus\{0\};
\end{equation}

\begin{equation}\label{J_0_ic}
J_{\mathfrak{r}}(x,t;k)=\begin{pmatrix}1&0\\\frac{z(k+0)}{k\,
a_+(k-0)a_+(k+0)}\,\e^{2\i g(y,t;k)}&1\end{pmatrix},\quad
k\in\(\i\sqrt{\frac{c}{4(c+\omega)}},0\);
\end{equation}

\begin{equation}\label{J_0_-ic}
J_{\mathfrak{r}}(x,t;k)=\begin{pmatrix}1&\dsfrac{z(k+0)}{k\,
\overline{a_+(\overline{k}-0)}\,\overline{a_+(\overline{k}+0)}}\,\e^{-2\i
g(y,t;k)}\\0&1\end{pmatrix},\quad
k\in\(0,-\i\sqrt{\dsfrac{c}{4(c+\omega)}}\);
\end{equation}

\item The pole relation: for $j=1,...,N$\\
\begin{subequations}\label{pole conditions 1}
\begin{eqnarray}
\mathrm{Res}_{\i\kappa_j}V_{\mathfrak{r}}(x,t;k)=\lim\limits_{k\rightarrow\i\kappa_j}V_{\mathfrak{r}}(x,t;k)\begin{pmatrix}0&0\\\i\gamma_{+,j}^2\e^{2\i
g(y,t;\i\kappa_j)}&0\end{pmatrix},\qquad\ \ \\
\nonumber\mathrm{Res}_{-\i\kappa_j}V_{\mathfrak{r}}(x,t;k)=\lim\limits_{k\rightarrow-\i\kappa_j}V_{\mathfrak{r}}(x,t;k)\begin{pmatrix}0&-\i\gamma_{+,j}^2\e^{2\i
g(y,t;\i\kappa_j)}\\0&0\end{pmatrix};\\
\end{eqnarray}
\end{subequations}

\item Symmetry relations:\\
\begin{equation}\label{symmetries M 1}
\overline{V_{\mathfrak{r}}(x,t;\overline{k})}=V_{\mathfrak{r}}(x,t;-k)=V_{\mathfrak{r}}(x,t;k)\begin{pmatrix}0&1\\1&0\end{pmatrix},
\end{equation}
\[\overline{V_{\mathfrak{r}}(x,t;-\overline{k})}=V_{\mathfrak{r}}(x,t;k);\]
 \item asymptotics at the infinity: \begin{equation}\label{M infinity}
V_{\mathfrak{r}}(x,t;k) \rightarrow
\begin{pmatrix}1&1\end{pmatrix} \textrm{ as }
k\rightarrow\infty.\end{equation}
\end{enumerate}

\textbf{Regular Riemann--Hilbert problem.} It is useful to
transform our meromorphic RH problem defined in section
\ref{Sect_RH_problem_1} to holomorphic RH problem. In order to
achieve this, we define function $\widehat
V_{\mathfrak{r}}(x,t;k)$ as follows:
\begin{equation}\label{hat M}
\widehat V_{\mathfrak{r}}(x,t;k)=\left\{
\begin{array}{l}
V_{\mathfrak{r}}(x,t;k)\begin{pmatrix}1&0\\\dsfrac{-\i\gamma_{+,j}^2\e^{2\i
g(y,t;\i\kappa_j)}}{k-\i\kappa_j}&1\end{pmatrix},\qquad
|k-\i\kappa_j|<\varepsilon,
\\
V_{\mathfrak{r}}(x,t;k)\begin{pmatrix}1&\dsfrac{\i\gamma_{+,j}^2\e^{2\i
g(y,t;\i\kappa_j)}}{k+\i\kappa_j}\\0&1\end{pmatrix}, \qquad
|k+\i\kappa_j|<\varepsilon,\\
V_{\mathfrak{r}}(x,t;k),\qquad \textrm{ elsewhere},
\end{array}
\right.
\end{equation}
where $\varepsilon>0$ is sufficiently small number such that
circles $|k-\i\kappa_j|$ do not intersect and lie in the domain
$\Im z(k)>0.$

\begin{lem}\label{lemma RH1 regular}
Vector-valued function $\widehat V_{\mathfrak{r}}(x,t;k)$ solves
the following RH problem: find a sectionally-holomorphic function
$\widehat V_{\mathfrak{r}}(x,t;k)$ which satisfies the following:
\begin{enumerate}
\item $\widehat V_{\mathfrak{r}}(x,t;.)$ is holomorphic away of
the contour $\Sigma_{\mathfrak{r}}$ and circuits $C_j:=\left\{k:\
|k-\i\kappa_j|=\varepsilon\right\},$ $\overline C_j:=\left\{k:\
|k+\i\kappa_j|=\varepsilon\right\}.$ The orientation on $C_j$,
$\overline C_j$ is counterclockwise, so the positive side is
inside the circuits. \item jump condition $\widehat
V_{\mathfrak{r}}^-(x,t;k)=\widehat
V_{\mathfrak{r}}^+(x,t;k)\widehat J_{\mathfrak{r}}(x,t;k)$ is
satisfied, where
\[\widehat J_{\mathfrak{r}}(x,t;k)\equiv J_{\mathfrak{r}}(x,t;k),\quad k\in\Sigma_{\mathfrak{r}},\]
\begin{equation}\label{J_hat Cj}
\widehat J_{\mathfrak{r}}(x,t;k)=
\begin{pmatrix}1&0\\\dsfrac{\i\gamma_{+,j}^2\e^{2\i
g(y,t;\i\kappa_j)}}{k-\i\kappa_j}&1\end{pmatrix} \quad k\in C_j;
\end{equation}
\begin{equation}\label{J_hat Cj overline}
\widehat J_{\mathfrak{r}}(x,t;k)=
\begin{pmatrix}1&\dsfrac{-\i\gamma_{+,j}^2\e^{2\i
g(y,t;\i\kappa_j)}}{k+\i\kappa_j}\\0&1\end{pmatrix} \quad k\in
\overline C_j;
\end{equation}
\item Symmetry relations:\\
\begin{equation}\label{symmetries M 1 hat}
\overline{\widehat V_{\mathfrak{r}}(x,t;\overline{k})}=\widehat
V_{\mathfrak{r}}(x,t;-k)=\widehat
V_{\mathfrak{r}}(x,t;k)\begin{pmatrix}0&1\\1&0\end{pmatrix},
\end{equation}
\[\overline{\widehat V_{\mathfrak{r}}(x,t;-\overline{k})}=\widehat V_{\mathfrak{r}}(x,t;k);\]
 \item $\widehat V_{\mathfrak{r}}(x,t;k)\rightarrow \begin{pmatrix}1&1\end{pmatrix}$ as
$k\rightarrow\infty.$
\end{enumerate}
\end{lem}

\section{Riemann--Hilbert problem (2)}\label{Sect_RH_problem_2}
Alternative Riemann-Hilbert problem can be stated for spectral
parameter $z$ in assumption that $l=1$ in (\ref{conditions
u(x,t)}). Let us define sectionally-meromorphic function

\[V_{\mathfrak{l}}(x,t;z)=\]
\begin{equation}\label{M_varphi_z}
=\left\{\begin{array}{l}
\sqrt[4]{\dsfrac{m+\omega}{\omega}}\(\dsfrac{1}{a_-(k(z))}\,\varphi_+(x,t;k(z))\e^{-\i
\widetilde g(y,t;z)}\ , \ \varphi_-(x,t;k(z))\e^{\i \widetilde
g(y,t;z)}\), \\\\\hfill \hfill\hfill\hfill\hfill\hfill \Im z>0,
\\
\\
\sqrt[4]{\dsfrac{m+\omega}{\omega}}\(\overline{\varphi_-(x,t;\overline{k(z)})}\e^{-\i
\widetilde g(y,t;z)}\ , \
\dsfrac{1}{\overline{a_-(\overline{k(z)})}}\
\overline{\varphi_+(x,t;\overline{k(z)})}\e^{\i \widetilde
g(y,t;z)}\),\\\\\hfill
\hfill\hfill\hfill\hfill\hfill\hfill\hfill\Im z<0,
\end{array}\right.
\end{equation}
\noindent where $\widetilde
g(y,t;z)=z(y+\wp)-\frac{2\sqrt{(c+\omega)\omega\ } z
t}{4z^2+\frac{\omega}{c+\omega}}.$

Scattering relations (\ref{scattering relations varphi}) and lemma
\ref{lemma properties a(k) b(k)} implies that
$V_{\mathfrak{l}}(x,t;z)$ satisfies the following conjugation
problem:

\begin{enumerate}
\item $V_{\mathfrak{l}}(x,t;.)$ is meromorphic in
$\mathbb{C}\setminus\mathbb{R}$ and continuous up to the boundary;
\item $V_{\mathfrak{l}}^-(x,t;z)=V_{\mathfrak{l}}^+(x,t;z)
J_{\mathfrak{l}}(x,t;z),$ where
\begin{equation}\nonumber
J_{\mathfrak{l}}(x,t;z)=\begin{pmatrix}1&\dsfrac{\overline{b_-(k(z))}}{\overline{a_-(k(z))}}\,\,\e^{2\i
\widetilde g(y,t;z)}\\-\dsfrac{b_-(k(z))}{a_-(k(z))}\,\e^{-2\i
\widetilde g(y,t;z)}&\dsfrac{k(z)}{z\,|a_-(k(z))|^2}\end{pmatrix},
\end{equation}
\begin{equation}\label{widetilde J_R}\hfill\quad
z\in\(-\infty,-\sqrt{\frac{c}{4(c+\omega)}}\)\cup\(\sqrt{\frac{c}{4(c+\omega)}},+\infty\);
\end{equation}

\begin{equation}\nonumber
J_{\mathfrak{l}}(x,t;z)=\begin{pmatrix}1&\dsfrac{\overline{b_-(k(z+\i
0))}}{\overline{a_-(k(z+\i 0))}}\,\,\e^{-2\i \widetilde
g(y,t;z)}\\-\dsfrac{b_-(k(z+\i 0))}{a_-(k(z+\i 0))}\,\e^{2\i
\widetilde g(y,t;z)}&0\end{pmatrix}, \end{equation}
\begin{equation}\label{widetilde J_0_ic}\hfill\quad z\in\(-\sqrt{\frac{c}{4(c+\omega)}},
\sqrt{\frac{c}{4(c+\omega)}}\);
\end{equation}
\item The pole relation: for $j=1,...,N$\\
\begin{subequations}\label{pole conditions 2}
\begin{eqnarray}
\mathrm{Res}_{\i
z_j}V_{\mathfrak{l}}(x,t;k)=\lim\limits_{z\rightarrow
z_j}V_{\mathfrak{l}}(x,t;z)\begin{pmatrix}0&0\\\i\gamma_{-,j}^2\e^{-2\i
\widetilde g(y,t;z_j)}&0\end{pmatrix},\qquad\ \ \\
\nonumber\mathrm{Res}_{-\i
z_j}V_{\mathfrak{l}}(x,t;z)=\lim\limits_{z\rightarrow-z_j}V_{\mathfrak{l}}(x,t;z)\begin{pmatrix}0&-\i\gamma_{-,j}^2\e^{-2\i
\widetilde g(y,t;z_j)}\\0&0\end{pmatrix};\\
\end{eqnarray}
\end{subequations}
\item Symmetry relations:\\
\begin{equation}\label{symmetries M 2}
\overline{V_{\mathfrak{l}}(x,t;\overline{z})}=V_{\mathfrak{l}}(x,t;-z)=V_{\mathfrak{l}}(x,t;z)\begin{pmatrix}0&1\\1&0\end{pmatrix},
\end{equation}
\[\overline{V_{\mathfrak{l}}(x,t;-\overline{z})}=V_{\mathfrak{l}}(x,t;z);\]

\item $V_{\mathfrak{l}}(x,t;z)\rightarrow
\begin{pmatrix}1&1\end{pmatrix} $ as $z\rightarrow\infty.$
\end{enumerate}

\section{Reconstruction of the Camassa--Holm solution from RH
problem}\label{sect reconstruction of potential}

The following lemma reads exactly as its analogue in vanishing
case (\cite[lemma 3.5]{Kostenko_Shepelsky_Teschl_2009}), but proof
should be modified.
\begin{lem}Functions $V_{\mathfrak{r}}(x,t;k)$ and $V_{\mathfrak{l}}(x,t;z)$ defined in
(\ref{M_psi})and (\ref{M_varphi_z}), respectively, satisfy the
following relations:
\begin{equation}\label{M1/M2}\dsfrac{V_{\mathfrak{r},\,1}(x,t;\frac{\i}{2})}{V_{\mathfrak{r},\,2}(x,t;\frac{i}{2})}=\e^{x-y},\end{equation}
\begin{equation}\label{M1*M2}V_{\mathfrak{r},\,1}(x,t;k)V_{\mathfrak{r},\,2}(x,t;k)=\sqrt{\frac{m+\omega}{\omega}}\(1+\dsfrac{2\i}{\omega}u(x,t)\(k-\frac{\i}{2}\)
+O\(k-\frac{\i}{2}\)^2\),
\end{equation}
\[\hfill k\rightarrow\frac{\i}{2},\]
and
\begin{equation}\label{M1/M2 tilde}\dsfrac{V_{\mathfrak{l},\,1}(x,t;\frac{\i}{2}\sqrt{\frac{\omega}{c+\omega}})}
{
V_{\mathfrak{l},\,2}(x,t;\frac{\i}{2}\sqrt{\frac{\omega}{c+\omega}})}=\e^{\sqrt{\frac{\omega}{c+\omega}}\(y+\wp\)-x+ct}=\exp\left\{
\int\limits_{-\infty}^x\(\sqrt{\frac{m+\omega}{c+\omega}}-1\)\d
r\right\},\end{equation}
\[V_{\mathfrak{l},\,1}(x,t;z)V_{\mathfrak{l},\,2}(x,t;z)=\]
\begin{equation}\label{M1*M2 tilde}=\sqrt{\frac{m+\omega}{c+\omega}}
\(1+\dsfrac{2\i
\(u(x,t)-c\)}{\sqrt{\omega(c+\omega)}}\(z-\frac{\i}{2}\sqrt{\frac{\omega}{c+\omega}}\)
+O\(z-\frac{\i}{2}\sqrt{\frac{\omega}{c+\omega}}\)^2\),\end{equation}
\[\hfill z\rightarrow\frac{\i}{2}\sqrt{\frac{\omega}{c+\omega}}.\]
\end{lem}

\begin{Proof}
The Jost solutions $\varphi_{\pm}(x,t;k)$ can be represented as
follows:
$$\varphi_+(x,t;k)=\mu_+(x,t;k)\exp\{\i
kx-\frac{2\i\omega k t}{4k^2+1}\},\ $$
$$\varphi_-(x,t;k)=\sqrt[4]{\frac{\omega}{c+\omega}}\mu_-(x,t;k)
\exp\left\{-\i\sqrt{\frac{c+\omega}{\omega}\ }x z+\i
\sqrt{\frac{c+\omega}{\omega}\ }t
z\(c+\frac{2\omega}{4k^2+1}\)\right\},$$

\noindent where $\mu_{\pm}(x,t;k)$ are the solutions of the
integral equations
\begin{equation}\label{mu_+_eq}
\mu_+(x,t;k)=1+\frac{k^2+\frac{1}{4}}{2\i
k\omega}\int\limits_x^{+\infty}\(1-\e^{-2\i
k(x-r)}\)\mu_+(r,t;k)m(r,t)\d r,
\end{equation}

\begin{equation}\label{mu_-_eq}
\mu_-(x,t;k)=1+\frac{k^2+\frac{1}{4}}{2\i z\sqrt{\omega(c+\omega)\
}}\int\limits^x_{-\infty}\(1-\e^{2\i
z(x-r)\sqrt{\frac{c+\omega}{\omega}\
}}\)\mu_-(r,t;k)\(m(r,t)-c\)\d r.
\end{equation}
Here the values of $\mu_{\pm}$ at the point $k=\frac{\i}{2}$ can
be determined due to the observation that
(\ref{Lax_representation_CH}a) becomes explicitly solvable at
$\lambda=0.$

Existence and uniqueness of solutions of
(\ref{mu_-_eq}),(\ref{mu_+_eq}) is established, for example, in
 \cite{Marchenko_1972_1986}.
Moreover, $\mu_+$ is analytic for $\Im k>0,\ $ $\mu_-$ is analytic
for $\Im z(k)>0.$

Since $k^2+\frac{1}{4}=(k-\frac{\i}{2})(k+\frac{\i}{2}),$ we get
\begin{equation}\label{mu_+_as_i/2}
\mu_+(x,t;k)=1+\frac{\i}{\omega}F_+(x,t)\(k-\frac{\i}{2}\)+O\(k-\frac{\i}{2}\)^2,
\quad k\rightarrow\frac{\i}{2},
\end{equation}

\begin{equation}\label{mu_-_as_i/2}
\mu_-(x,t;k)=1-\frac{\i}{\omega}F_-(x,t)\(k-\frac{\i}{2}\)+O\(k-\frac{\i}{2}\)^2,
\quad k\rightarrow\frac{\i}{2},
\end{equation}

where
\begin{equation}\label{F_+(x,t)}
F_+(x,t)=\int\limits_x^{+\infty}\(\e^{x-r}-1\)m(r,t)\d r,
\end{equation}
\begin{equation}\label{F_-(x,t)}
F_-(x,t)=\int\limits^x_{-\infty}\(\e^{-x+r}-1\)\(c-m(r,t)\)\d r.
\end{equation}

Differentiating with respect to $x$, we get

\begin{equation}\label{mu'_+_as_i/2}
\mu'_+(x,t;k)=\frac{\i}{\omega}F'_+(x,t)\(k-\frac{\i}{2}\)+O\(k-\frac{\i}{2}\)^2,
\quad k\rightarrow\frac{\i}{2},
\end{equation}

\begin{equation}\label{mu'_-_as_i/2}
\mu'_-(x,t;k)=-\frac{\i}{\omega}F'_-(x,t)\(k-\frac{\i}{2}\)+O\(k-\frac{\i}{2}\)^2,
\quad k\rightarrow\frac{\i}{2},
\end{equation}

with
\begin{equation}\label{F'_+(x,t)}
F'_+(x,t)=\int\limits_x^{+\infty}\e^{x-r}m(r,t)\d r,
\end{equation}
\begin{equation}\label{F'_-(x,t)}
F'_-(x,t)=-\int\limits^x_{-\infty}\e^{-x+r}\(c-m(r,t)\)\d r.
\end{equation}

Next, straightforward calculations show that
\[a_+(k)=\frac{W\{\varphi_-(x,t;k),\varphi_+(x,t;k)\}}{2\i k}=
\sqrt[4]{\frac{\omega}{c+\omega}}\(1+\right.\]\[\left.+\dsfrac{\i}{\omega}\left[F_+-F_--F'_+-F'_--c(x+1)+ct\(\frac{3c}{2}+2\omega\)\right]
\(k-\frac{\i}{2}\)+\mathrm{O} \(k-\frac{\i}{2}\)^2\)=\]
\begin{equation}\label{a_+_as_i/2}=\sqrt[4]{\frac{\omega}{c+\omega}}\left\{1+\frac{\i}{\omega}H_0[u]
\(k-\frac{\i}{2}\)+O\(k-\frac{\i}{2}\)^2\right\},\end{equation}
where
\begin{equation}\label{H_0[u]}
H_0[u]=\int_{-\infty}^x \(c-m(r,t)\)\d
r-\int\limits_x^{+\infty}m(r,t)\d
r-c(x+1)+\frac{c(3c+4\omega)t}{2}
\end{equation}
is a conserved quantity of the CH equation.

Substituting (\ref{mu_-_as_i/2}), (\ref{mu_+_as_i/2}) and
(\ref{a_+_as_i/2}) into (\ref{M_psi}), and taking into account
\[u(x,t)=(1-\partial_x^2)^{-1}m(x,t)=\frac{1}{2}\int\limits_{\mathbb{R}}\e^{-|x-r|}m(r,t)\d r,\]
we come to
\[V_{\mathfrak{r},\,1}(x,t;k)=\sqrt[4]{\frac{m+\omega}{\omega}\ }\e^{(x-y)/2\
}\(1+\left[
\frac{-\i}{\omega}H_0[u]-\frac{\i}{\omega}F_-(x,t)+\right.\right.\]
\[\left.\left.\qquad-\i x\frac{c+\omega}{\omega}+\i y+\frac{\i ct(3c+4\omega)}{2\omega}\right]\(k-\frac{\i}{2}\)
+O\(k-\frac{\i}{2}\)^2 \),\]

\[V_{\mathfrak{r},\,2}(x,t;k)=\]\[=\sqrt[4]{\frac{m+\omega}{\omega}}\e^{(y-x)/2}\(1+\(\frac{\i}{\omega}F_+(x,t)+\i(x-y)\)\(k-\frac{\i}{2}\)+O\(k-\frac{\i}{2}\)^2\),\]
and thus to (\ref{M1/M2}), (\ref{M1*M2}). Formulas (\ref{M1/M2
tilde}), (\ref{M1*M2 tilde}) can be proved in analogous manner.
\end{Proof}

\section{Uniqueness.}\subsection{Uniqueness of RH problems by the right scattering data.}
In this section we prove the uniqueness of the RH problems 1.-5.
of the section \ref{Sect_RH_problem_1} under assumtion that the
classical  solution of the Cauchy problem (\ref{CH}),
(\ref{init_cond}), (\ref{conditions u(x,t)}), (\ref{m+omega})
exists for all values of time. Then we can construct the
matrix-valued function
$M_{\mathfrak{r}}(x,t;k)=\begin{pmatrix}M_{\mathfrak{r}
,11}(x,t;k)&M_{\mathfrak{r} ,12}(x,t;k)\\M_{\mathfrak{r}
,21}(x,t;k)&M_{\mathfrak{r} ,22}(x,t;k)\end{pmatrix}$ by the
formulas
\begin{equation}\label{M_matrix_right_upper}
\begin{pmatrix}
\frac{1}{2a_+(k)}\(\sqrt[4]{\frac{m+\omega}{\omega}}\varphi_--\frac{1}{\i
k}\sqrt[4]{\frac{\omega}{m+\omega}}\varphi_{-,x}'\)\e^{\i g}&
\frac{1}{2}\(\sqrt[4]{\frac{m+\omega}{\omega}}\varphi_+-\frac{1}{\i
k}\sqrt[4]{\frac{\omega}{m+\omega}}\varphi_{+,x}'\)\e^{-\i g}
\\
\frac{1}{2a_+(k)}\(\sqrt[4]{\frac{m+\omega}{\omega}}\varphi_-+\frac{1}{\i
k}\sqrt[4]{\frac{\omega}{m+\omega}}\varphi_{-,x}'\)\e^{\i g}&
\frac{1}{2}\(\sqrt[4]{\frac{m+\omega}{\omega}}\varphi_++\frac{1}{\i
k}\sqrt[4]{\frac{\omega}{m+\omega}}\varphi_{+,x}'\)\e^{-\i g}
\nonumber\end{pmatrix},
\end{equation}
\begin{equation}\hfill\hfill\hfill\hfill \Im z(k)>0,
\end{equation}
\begin{equation}\label{M_matrix_right_lower}
\begin{pmatrix}
\frac{1}{2}\(\sqrt[4]{\frac{m+\omega}{\omega}}\ol{\varphi_+}-\frac{1}{\i
k}\sqrt[4]{\frac{\omega}{m+\omega}}\ol{\varphi_{+,x}'}\)\e^{\i g}
&\frac{1}{2\ol{a_+(\ol
k)}}\(\sqrt[4]{\frac{m+\omega}{\omega}}\ol{\varphi_-}-\frac{1}{\i
k}\sqrt[4]{\frac{\omega}{m+\omega}}\ol{\varphi_{-,x}'}\)\e^{-\i g}
\\
\frac{1}{2}\(\sqrt[4]{\frac{m+\omega}{\omega}}\ol{\varphi_+}+\frac{1}{\i
k}\sqrt[4]{\frac{\omega}{m+\omega}}\ol{\varphi_{+,x}'}\)\e^{\i g}
&\frac{1}{2\ol{a_+(\ol
k)}}\(\sqrt[4]{\frac{m+\omega}{\omega}}\ol{\varphi_-}+\frac{1}{\i
k}\sqrt[4]{\frac{\omega}{m+\omega}}\ol{\varphi_{-,x}'}\)\e^{-\i g}
\nonumber\end{pmatrix},
\end{equation}
\begin{equation}\hfill\hfill\hfill\hfill \Im z(k)<0,
\end{equation}

Here we denote $g=g(y,t;k)=k y-\frac{2\omega k t}{4k^2+1},$
$\varphi_{\pm}=\varphi_{\pm}(x,t;k),$
$\ol{\varphi_{\pm}}=\ol{\varphi_{\pm}(x,t;\ol k)}.$ Matrix-valued
function $M_{\mathfrak{r}}$ has the following properties:
\begin{enumerate}[1.]
\item Matrix-valued function $M_{\mathfrak{r}}(x,t;k)$ is
meromorphic in the domain
$\mathbb{C}\backslash\Sigma_{\mathfrak{r}}.$
\end{enumerate}
\begin{enumerate}[2.]\item Matrix-valued function $M_{\mathfrak{r}}(x,t;k)$ satisfies
the same jump relations as the vector-valued function
(\ref{M_psi}).
\end{enumerate}
\begin{enumerate}[3.]
\item $M_{\mathfrak{r}}$ satisfies the same pole relations as the
vector-valued function (\ref{M_psi}).
\end{enumerate}
Indeed, the derivatives of the Jost solutions $\varphi_{\pm}'$
satisfy the same scattering relations (\ref{scattering relations
varphi}) as the Jost solutions. The function $k$ does not have any
jump along the contour $\Sigma_{\mathfrak{r}},$ hence any linear
combination of Jost solutions and its derivatives satisfies the
same jump conditions as the vector-valued function (\ref{M_psi}),
and in particular, this holds for $M_{\mathfrak{r}}(x,t;k)$.

Asymptotics for the derivatives of the Jost solutions is given by
formal differentiation of formulas (\ref{asymptotics varphi - k}),
(\ref{asymptotics varphi + k}), so the matrix-valued function
$M_{\mathfrak{r}}(x,t;k)$ tends to the identity matrix as
$k\rightarrow\infty.$


$M_{\mathfrak{r}}$ has the following behavior at the point $k=0$:
\begin{enumerate}[1.]
\item\hskip-2mm a $\begin{pmatrix}\mathrm{O}\(1\)&
\frac{-\alpha_{\pm}}{\i
k}+\mathrm{O}\(1\)\\
\mathrm{O}\(1\)&\frac{\alpha_{\pm}}{\i
k}+\mathrm{O}\(1\)\end{pmatrix},\quad
k\rightarrow0, \Im k>0, \pm\Re k\geq0;$\\
$\begin{pmatrix}\frac{-\alpha_{\pm}}{\i
k}+\mathrm{O}\(1\)&\mathrm{O}\(1\)\\
\frac{\alpha_{\pm}}{\i
k}+\mathrm{O}\(1\)&\mathrm{O}\(1\)\end{pmatrix},\quad
k\rightarrow0, \Im k<0, \pm\Re k\geq0;$
\end{enumerate}
 and satisfies the
following symmetry relations and asymptotics at the infinity:
\begin{enumerate}
\setcounter{enumi}{3}
\item \hskip-2mm(a) Symmetry relations:\\
\begin{equation}\label{symmetries M 1 matrix}
\overline{
M_{\mathfrak{r}}(x,t;\overline{k})}=M_{\mathfrak{r}}(x,t;-k)=
\begin{pmatrix}0&1\\1&0\end{pmatrix}M_{\mathfrak{r}}(x,t;k)\begin{pmatrix}0&1\\1&0\end{pmatrix},
\end{equation}
\[\overline{M_{\mathfrak{r}}(x,t;-\overline{k})}=M_{\mathfrak{r}}(x,t;k);\]
 \item \hskip-2mm(a) \begin{equation}\label{asymptotics M 1 matrix} M_{\mathfrak{r}}(x,t;k)\rightarrow \begin{pmatrix}1&0\\0&1\end{pmatrix}
 \textrm{ as
} k\rightarrow\infty.\end{equation}
\end{enumerate}

It is easy to show that the solution of the RH problem 1.-5. is
unique. (Let us notice, that this RH problem does not have
Schwartz symmetry:
$J_{\mathfrak{r}}(k)J^*_{\mathfrak{r}}(\ol{k})\neq I.$ This makes
possible the uniqueness of a non-regular solution.)First of all,
without loss of generality we can assume absence of poles at
$\pm\i \kappa_j$, because we can transform the RH problem to a
regular one as in (\ref{hat M}). Further, as $\det
J_{\mathfrak{r}}=1$, $det M_{\mathfrak{r}}$ is a meromorphic
function with a possible simple pole at the point $k=0$. As
$k\rightarrow\infty$, $\det M_{\mathfrak{r}}\rightarrow 1.$ Hence,
$\det M_{\mathfrak{r}}=1+\frac{b}{k}.$ Due to symmetry condition
4(a) we get that $\det M_{\mathfrak{r}}(-k)=\det
M_{\mathfrak{r}}(k)$, and hence $b=0$, so $\det
M_{\mathfrak{r}}=1$. If the RH problem 1.-5. has another solution
$\widetilde M_{\mathfrak{r}},$ then $M_{\mathfrak{r}}{\widetilde
M_{\mathfrak{r}}}^{-1}$ can have only simple pole at the origin:
it follows from the structure 1.a of poles at the origin. So,
$M_{\mathfrak{r}}{\widetilde M_{\mathfrak{r}}}^{-1}$ is a
meromorphic matrix-valued function with a possible simple pole at
the point $k=0$, hence $M_{\mathfrak{r}}{\widetilde
M_{\mathfrak{r}}}^{-1}=I+\frac{A}{\i k}$, where
$A=\begin{pmatrix}A_{11}&A_{12}\\A_{21}&A_{22}\end{pmatrix}.$ Due
to the symmetry properties 4(a), we get that the matrix $A$ has
the following structure: $A_{ij}\in\mathbb{R}, A_{11}=-A_{22},
A_{21}=-A_{12}.$ The determinant of $I+\frac{A}{\i k}$ equals 1,
this gives $A_{11}^2+A_{21}^2=0$ and so it proves the uniqueness
of the matrix RH problem 1.-5.

With the help of this matrix function $M_{\mathfrak{r}}$ we can
prove the uniqueness of the vector RH problem 1.-5. from the
section \ref{Sect_RH_problem_1}. Indeed, as $V_{\mathfrak{r}}$ and
$M_{\mathfrak{r}}$ satisfy the same jump conditions, their
quotient $V_{\mathfrak{r}} M_{\mathfrak{r}}^{-1}$ does not have
any jump along the contour $\Sigma_{\mathfrak{r}}$, and therefore
$V_{\mathfrak{r}} M_{\mathfrak{r}}^{-1}$ is a meromorphic
vector-valued function with a possible simple pole at the point
$k=0$ and asymptotics at the infinity equals $(1,1)$. Due to the
symmetry relation 4. of the section \ref{Sect_RH_problem_1} and
due to the symmetry relation 4.a of the present section, we get
that $V_{\mathfrak{r}} M_{\mathfrak{r}}^{-1}$ satisfies the
symmetry relations 4. of the section \ref{Sect_RH_problem_1}. So,
\begin{equation}\label{VM^-1}V_{\mathfrak{r}} M_{\mathfrak{r}}^{-1}=\(1+\frac{A}{\i k},1-\frac{A}{\i k}\).
\end{equation}

According to the property 1.a, the function
$M_{\mathfrak{r}}(x,t;k)$ has the following behavior as
$k\rightarrow0,$ $\Im k>0$, $\Re k>0$:
\begin{equation}\label{M_r k=0}M_{\mathfrak{r}}(k)=\begin{pmatrix}\beta+\mathrm{O}(k)&\frac{-\alpha}{\i
k}+x+\mathrm{O}(k)\\
\gamma+\mathrm{O}(k)&\frac{\alpha}{\i
k}+y+\mathrm{O}(k)\end{pmatrix},\end{equation} where
$\alpha\in\mathbb{R}.$ From the fact that $\det
M_{\mathfrak{r}}(x,t;k)=1$, we get
\begin{equation}\label{M_r eq_1}
\alpha(\beta+\gamma)=0.
\end{equation}
Let us multiply (\ref{VM^-1}) from the right by (\ref{M_r k=0}).
We have
$$V_{\mathfrak{r}}=\(V_{\mathfrak{r}} M_{\mathfrak{r}}^{-1}\)M_{\mathfrak{r}}(k)=\(1+\frac{A}{\i
k},1-\frac{A}{\i
k}\)\begin{pmatrix}\beta+\mathrm{O}(k)&\frac{-\alpha}{\i
k}+x+\mathrm{O}(k)\\
\gamma+\mathrm{O}(k)&\frac{\alpha}{\i
k}+y+\mathrm{O}(k)\end{pmatrix}=$$
$$=\begin{pmatrix}
\frac{A(\beta-\gamma)}{\i k}+\mathrm{O}(1),
\frac{2A\alpha}{k^2}+\frac{A(x-y)}{\i
k}+\mathrm{O}(1)\end{pmatrix},$$ and, taking into account
boundedness of $V_{\mathfrak{r}}$ as $k\rightarrow 0$, we have
\begin{equation}\label{M_r eq_2}
A(\beta-\gamma)=0,\quad A\alpha=0,\quad A(x-y)=0.
\end{equation}

Consider the system of equations (\ref{M_r eq_1}), (\ref{M_r
eq_2}). Suppose first that $A\neq0.$ Then $$\alpha=0,\quad
\beta=\gamma, \quad x=y,$$ and formula (\ref{M_r k=0}) reads as
$$M_{\mathfrak{r}}(k)=\begin{pmatrix}\beta+\mathrm{O}(k)&x+\mathrm{O}(k)\\
\beta+\mathrm{O}(k)&x+\mathrm{O}(k)\end{pmatrix}, \quad\textrm{
hence } \det M_{\mathfrak{r}}(k)=\mathrm{O}(k)\neq1.$$ We get a
contradiction. It proves that the matrix $A=0$. From this and from
(\ref{VM^-1}) we get
$$V_{\mathfrak{r}}M_{\mathfrak{r}}^{-1}=(1,1).$$

If $\widetilde V_{\mathfrak{r}}(x,t;k)$ is another solution of the
vector RH problem from the section \ref{Sect_RH_problem_1}, then
$$V_{\mathfrak{r}}M_{\mathfrak{r}}^{-1}=(1,1)=\widetilde{V_{\mathfrak{r}}}M_{\mathfrak{r}}^{-1}$$
and, taking into account $\det M_{\mathfrak{r}}=1$, we have
$V_{\mathfrak{r}}=\widetilde{V_{\mathfrak{r}}},$ which is what had
to be proved.

\subsection{Uniqueness of general vector RH problem.}

As to the RH problem stated by the left scattering data, we can
repeat the procedure of the previous subsection, and construct the
matrix-valued function $M_{\mathfrak{l}}$. The distinction is that
pole behavior at $z=0$ involves all 4 elements of the
corresponding matrix RH problem solution. It provides
non-uniqueness of the non-regular matrix RH problem solution.
Indeed, the RH problem now is stated on the real line, and
satisfies Schwartz symmetry principle; this means existence of a
regular solution to the RH problem. However, we still have the
uniqueness of the vector RH problem. This is a general fact, and
we prove it in a general form.

\begin{lem}\label{lemma uniqueness RH problem}
Suppose $\Sigma\in\mathbb{C}$ is a smooth contour, with possible
self-intersections. The following RH problem has a unique
solution: to find a vector-valued function $V(x,t;z)$ that
satisfies the following properties
\begin{enumerate}
\item $V(y,t;.)$ is holomorphic in $\mathbb{C}\setminus\Sigma$ and
continuous up to the boundary; \item $V^-(x,t;z)=V^+(x,t;z)
J(y,t;z), z\in \Sigma,$ where
\begin{equation}\nonumber
J(y,t;z)=\e^{-\i z y\sigma_3}J_0(t;z)\e^{\i z
y\sigma_3}\equiv\e^{-\i z
y\sigma_3}\begin{pmatrix}J_{11}(t;z)&J_{12}(t;z)\\J_{21}(t;z)&J_{22}(t;z)\end{pmatrix}\e^{\i
z y\sigma_3}.
\end{equation}
$J_0(t;z)$ is a continuous function on the contour $\Sigma$ and
$\det J_0(t;z)\equiv 1;$
\item Symmetry relations:\\
\begin{equation}\label{symmetries M 2}
\overline{V(y,t;\overline{z})}=V(y,t;-z)=V(y,t;z)\begin{pmatrix}0&1\\1&0\end{pmatrix},
\end{equation}
\[\overline{V(y,t;-\overline{z})}=V(y,t;z);\]

\item $V(y,t;z)\rightarrow
\begin{pmatrix}1&1\end{pmatrix} $ as $z\rightarrow\infty.$
\end{enumerate}

\end{lem}

\begin{proof}
Suppose $V(y,t;z)=(V_1(y,t;z),V_2(y,t;z))$ satisfies the RH
problem stated above. The jump relation 2
\begin{equation}\nonumber
V_1^-=V_1^+ J_{11}+V_2^+ \e^{2\i y z} J_{21},\qquad V_2^-=V_1^+
\e^{-2\i y z}J_{12}+V_2^+ J_{22},
\end{equation}
can be rewritten as analogues of scattering relations between Jost
solutions:
\begin{equation}\label{relations V 1 2}
V_1^- \e^{-\i y z}=V_1^+\e^{-\i y z} J_{11}+V_2^+ \e^{\i y
z}J_{21},\qquad V_2^-\e^{\i y z}=V_1^+ \e^{-\i y
z}J_{12}+V_2^+\e^{\i y z} J_{22},
\end{equation}
Differentiating (\ref{relations V 1 2}) by $y$, we get

\begin{eqnarray}\label{relations V derivatives 1 2}
(V_1^- \e^{-\i y z})_y=(V_1^+\e^{-\i y z})_y J_{11}+(V_2^+ \e^{\i
y z})_y J_{21},\\ (V_2^-\e^{\i y z})_y=(V_1^+ \e^{-\i y z})_y
J_{12}+(V_2^+\e^{\i y z})_y J_{22}.
\end{eqnarray}
Dividing the last two equations by $\i z$, we get that functions
$\frac{1}{\pm \i z}(V_{1,2}\e^{\pm i yz})_y$ satisfy the same jump
relations and asymptotics for large $z$ as the functions
$V_{1,2}\e^{\pm i yz}$. Taking then half-sums and half-difference
between these functions, using (\ref{relations V derivatives 1 2})
and (\ref{relations V 1 2}), and then by multiplying them by
$\e^{\pm\i y z}$, we organize these functions as a matrix
\begin{equation}{\label{V to M matrix}}
M(y,t;z)=\begin{pmatrix}V_1-\frac{1}{2\i z}(V_1)_y & -\frac{1}{2\i
z}(V_2)_y\\
\frac{1}{2\i z}(V_1)_y & V_2-\frac{1}{2\i z}(V_1)_y \end{pmatrix}.
\end{equation}

Function $M(y,t;z)$ satisfies the same jump conditions as
$V(y,t;z),$ symmetry conditions (\ref{symmetries M 1 matrix}), and
asymptotics (\ref{asymptotics M 1 matrix}). As $\det J=1,$ the
determinant of $M$ does not have any jump across $\Sigma$.
Furthermore, it is a meromorphic function with at most a simple
pole at $z=0$: indeed, it equals \begin{equation}\label{det M}\det
M=\dsfrac{V_{1}V_{2,\,y}-V_{2}V_{1,\,y}}{2\i
z}+V_1V_2.\end{equation} Due to symmetry properties
(\ref{symmetries M 1 matrix}) it $\det M$ be represented as $$\det
M=1+\dsfrac{b}{\i z}, \quad b(y,t)\in\mathbb{R}.$$ But due to
properties (\ref{symmetries M 1 matrix}) we conclude that $\det
M(z)=\det M(-z),$ so $b\equiv 0$ and $\det M=1.$ Due to (\ref{det
M}) we get
\begin{equation}\label{det M a}\dsfrac{V_{1}V_{2,\,y}-V_{2}V_{1,\,y}}{2\i
z}+V_1V_2=1.\end{equation}

Suppose $\widetilde V=(\widetilde V_1, \widetilde V_2)$ is another
solutions of the vector RH problem stated in Lemma \ref{lemma
uniqueness RH problem}. Then by dividing it from the right on $M$,
we see that the quotient does not have any jump across $\Sigma$,
so it meromorphic with at most a simple pole at $z=0,$
$$\widetilde V M^{-1}=\(1-\frac{b}{\i z}, 1+\frac{b}{\i z}\),\quad b=b(x,t)\in\mathbb{R}.$$

Next, multiply the above relationship by $M$ from the right; we
get
\begin{equation}(\widetilde V_1, \widetilde V_2)=\(1-\frac{b}{\i z}, 1+\frac{b}{\i z}\)
\begin{pmatrix}V_1-\frac{1}{2\i z}V_{1,\,y} & -\frac{1}{2\i
z}V_{2,\,y}\\
\frac{1}{2\i z}V_{1,\,y} & V_{2}-\frac{1}{2\i z}V_{1,\,y}
\end{pmatrix}=
\end{equation}
$$=\(\dsfrac{b V_{1,\,y}}{-z^2}+\dsfrac{bV_1}{-\i z}+V_1,
\dsfrac{b V_{2,\,y}}{-z^2}+\dsfrac{bV_2}{\i z}+V_2\).$$ We recall
that $\widetilde V$ is a bounded function at $z=0$. Suppose first
that $b\neq 0.$ Then the vector $$\(\dsfrac{
V_{1,\,y}}{-z^2}+\dsfrac{V_1}{-\i z}, \dsfrac{
V_{2,\,y}}{-z^2}+\dsfrac{V_2}{\i z}\)$$ is also bounded at $z=0.$
Multiply the first element by $V_2$, the second by $-V_1$, and add
them. By using (\ref{det M a}), we get $$ -V_2\(\dsfrac{
V_{1,\,y}}{-z^2}+\dsfrac{V_1}{-\i z}\)+V_1\(\dsfrac{
V_{2,\,y}}{-z^2}+\dsfrac{V_2}{\i z}\)=\dsfrac{-2}{\i z},$$ we come
to contradiction, that is why $b\equiv 0.$ So, $\widetilde
V=(1,1)M=V.$ That is what had to be proved.

\end{proof}

\textbf{Acknowledgment.} The research has been supported by the
project "Support of inter-sectoral mobility of and quality
enhancement of research teams at Czech Technical University in
Prague", CZ.1.07/2.3.00/30.0034, sponsored by European Social Fund
in the Czech Republic.

The author would like to express his gratitude to I.Egorova for
useful advices concerning the direction of investigations, to
D.Shepelsky for introduction in the subject of the Camassa -- Holm
equation.

\bibliographystyle{elsarticle-num}
\bibliography{<your-bib-database>}







\end{document}